\documentclass[journal]{IEEEtran}
\usepackage{amsmath,mathtools,amsfonts,amssymb,amsthm}
\usepackage{amsmath,amsfonts}
\usepackage{algpseudocode}
\usepackage{algorithm}
\usepackage{array}
\usepackage[caption=false,font=normalsize,labelfont=sf,textfont=sf]{subfig}
\usepackage{textcomp}
\usepackage{stfloats}
\usepackage{url}
\usepackage{verbatim}
\usepackage{graphicx}
\usepackage{soul}
\def\BibTeX{{\rm B\kern-.05em{\sc i\kern-.025em b}\kern-.08em
    T\kern-.1667em\lower.7ex\hbox{E}\kern-.125emX}}
\usepackage{balance}
\usepackage{multirow}
\usepackage{booktabs} 

\usepackage{caption}                         
\usepackage{algorithm}
\usepackage{algpseudocode}
\usepackage{tablefootnote}
\usepackage[numbers,sort&compress]{natbib}

\newtheoremstyle{sltheorem}
{}                
{}                
{}        
{10pt}                
{\bfseries}       
{:}               
{ }               
{}                
\theoremstyle{sltheorem}
\newtheorem{assumption}{Assumption}
\newtheorem{theorem}{Theorem}
\newtheorem{definition}{Definition}
\newtheorem{corollary}{Corollary}

\begin{document}

\title{The Transmission Value of Energy Storage and Fundamental Limitations}
\author{Qian Zhang,~\IEEEmembership{Student Member, IEEE,} P. R. Kumar,~\IEEEmembership{Life Fellow, IEEE}, and Le Xie,~\IEEEmembership{Fellow, IEEE} 
\thanks{\indent The authors are with the Department of Electrical and Computer Engineering, Texas A\&M University, College Station, TX, 77843 USA (e-mail: zhangqianleo@tamu.edu; prk@tamu.edu; le.xie@tamu.edu).}}

\markboth{Journal of \LaTeX\ Class Files,~Vol.~18, No.~9, September~2020}%
{How to Use the IEEEtran \LaTeX \ Templates}

\maketitle

\begin{abstract}
This study addresses the transmission value of energy storage in electric grids. The inherent connection between storage and transmission infrastructure is captured from a ``cumulative energy" perspective, which enables the reformulating of the conventional optimization problem by employing line power flow as the decision variable. The study also establishes the theoretical limitations of both storage and transmission lines that can be replaced by each other, providing explicit closed-form expressions for the minimum capacity needed. As a key departure from conventional practice in which transmission lines are designed according to the peak power delivery needs, with sufficient storage capacity, the transmission line capacity can be designed based on the \emph{average} power delivery needs. The models of this paper only rely on a few basic assumptions, paving the way for understanding future storage as a transmission asset market design. Numerical experiments based on 2-bus, modified RTS 24-bus, RTS-GMLC, and Texas synthetic power systems illustrate the results.
\end{abstract}


\section{Introduction}
The last decade has witnessed remarkable developments in grid-scale energy storage, both in technological innovation and the participation of multiple electricity markets. On the technological front, reduced construction and operational costs, amplified capacity, higher energy density, and extended duration characterize the trajectory of various energy storage types, including chemical batteries, pumped storage hydro (PSH), and compressed-air energy storage (CAES) \cite{schmidt2019projecting,viswanathan20222022}. Notably, the proliferation of stationary battery storage has become more pronounced in recent years, attributed to its adaptable sizing, ease of installation, and breakthroughs in materials science, such as Li-ion iron phosphate (LFP), Li-ion nickel manganese cobalt (NMC), Lead-acid, etc. While currently battery storage may be economically inefficient for large-capacity or long-duration applications compared to PSH and CAES, ongoing efforts aim to enhance grid-scale battery storage duration and cycles \cite{chen2018manganese,zhu2022rechargeable}.\\
\indent Simultaneously, diverse market mechanisms for enhancing the usage of energy storage are also evolving across various electricity markets in the world. Beyond conventional arbitrage revenue from the energy market, new markets for frequency regulation, voltage support, and reserve capacity are emerging \cite{shi2017using,krata2018real,zheng2023energy}. This expanded range of market choices empowers storage owners to selectively participate in specific scenarios, providing an increasingly appealing prospect for investors \cite{stephan2016limiting,braeuer2019battery}. \\
\indent Transmission lines, as critical elements of the power system, are essential for conveying energy from supply sources to demand centers. Despite the deregulated electricity market implementing new products, like Financial Transmission Rights (FTRs), to encourage investment in transmission lines, the pace of transmission capacity expansion is still not enough for rapid renewable energy integration, resulting in significant curtailment and congestion issues \cite{millstein2021solar}. On the other hand, since transmission infrastructure is traditionally designed to meet peak demand and provide emergency backup, it often goes underutilized most of the time.\\
\indent The process of constructing new high-voltage transmission lines in the United States is lengthy, often exceeding 10 years, due to the complex permit processes and the challenge of cost allocation among numerous stakeholders. In response, several recent technological advances aim to enhance transmission capacity by improving the efficiency of existing lines. These advancements, known as Grid Enhancing Technologies (GETs), include innovations like dynamic line ratings and topology optimization, as highlighted in recent studies and reports \cite{USDOE2022GET,abboud2022guide}.\\
\indent Recently, the concept of using storage to increase the utilization of transmission assets has gained substantial attention, signifying a paradigm shift in the role of energy storage \cite{twitchell2022enabling,kumaraswamy2019redrawing, California,pjm}. However, the precise relationship and limitations between transmission lines and energy storage remain ambiguous, creating barriers to fully integrating storage as a transmission asset in the market. \\
\indent The interplay between energy storage and transmission lines was initially examined through joint optimization problems that considered the expansion of both storage and transmission capabilities, employing stochastic and multistage modeling approaches \cite{pandvzic2014near, singh2015storage, qiu2016stochastic, dvorkin2017co}. Despite these efforts, the variability in assumptions and system parameter configurations makes it challenging to derive any universally applicable conclusions from their sensitivity analysis results. For example, the simulation in \cite{qiu2016stochastic} shows the transmission capacity has an almost negligible effect on the total available energy storage capacities. In contrast, the transmission expansion may reduce storage penetration after considering different capital cost scenarios of storage in \cite{dvorkin2017co}. Additionally, incorporating contingency analysis into these co-optimization schemes presents a significant challenge \cite{twitchell2022enabling}.\\
\indent We aim to uncover the intrinsic relationship between energy storage and transmission infrastructure, particularly highlighting the transmission value of energy storage alongside its fundamental limitations. The model for energy storage and transmission lines presented in this paper is deliberately simplified to prevent narrowly focused results. We eschew focusing on particular systems and technologies. However, care has been taken to preserve the distinctive characteristics that set them apart from other elements within the power grid. The contributions of this work are the following:
{\begin{enumerate}
    \item Enhancing the understanding of the transmission value of storage by viewing the conventional optimization problem from a ``cumulative energy" perspective, defined in the sequel.
    \item Establishing the fundamental limitations of both transmission lines and storage. 
    \item Providing an explicit closed-form expression of the minimum storage capacity and the minimum line capacity needed under these limits, respectively.
    \item Proposing a tractable contingency analysis method to select critical scenarios.
\end{enumerate}} 
\indent We specifically defer future market design considerations of storage as a transmission asset \cite{khastieva2019value,brown2020motivating} to our forthcoming work. The remainder of this paper is organized as follows: Section~\ref{sec:Problem} describes the modeling assumptions and the conventional optimization problem. The main results of this paper are concentrated in Section \ref{sec:res}, including the reformulated problem under the cumulative energy perspective and the fundamental limitations. In Section \ref{sec:dis}, some practical issues like scenario generation and contingency analysis are discussed. We illustrate the transmission value of energy storage and demonstrate the correctness of the fundamental limits theorem on two-bus, modified RTS 24-bus, RTS-GMLC, and Texas synthetic power systems in Section~\ref{sec:case}.\\
\indent Some key symbols used are listed in TABLE \ref{symbol}.
\begin{table}[H] 
\caption{Notation}
\scalebox{0.9}{
\begin{tabular}{ll} 
\toprule
\multicolumn{1}{c}{Symbol}     & \multicolumn{1}{c}{Description}                                         \\ \midrule
\multirow{2}{*}{$\alpha_{ij}$} & \multirow{2}{*}{\begin{tabular}[c]{@{}l@{}}Discounted installation cost per unit of power for\\ transmission line connecting $i$ and $j$\end{tabular}} \\
                               &                                                                                                                                              \\
$\beta_i$              & Discounted installation cost per unit of storage at bus $i$                                \\             $\gamma^+_i$           & Cost per unit of generation curtailment at bus $i$                              \\
$\gamma^-_i$           & Cost per unit of load shedding at bus $i$   \\  
$C_{ij}=C_{ji}$               & Existing capacity of transmission line connecting $i$ and $j$       \\
$c_{ij}=c_{ji}$               & The capacity of new transmission line connecting $i$ and $j$                   \\
\multirow{2}{*}{$E_i[t]$} & \multirow{2}{*}{\begin{tabular}[c]{@{}l@{}}Total energy produced up to time slot $t$ at bus $i$\\  \emph{before} curtailment or load shedding\end{tabular}} \\
                               &                                 \\
\multirow{2}{*}{$\underline{E}_i[t]$} & \multirow{2}{*}{\begin{tabular}[c]{@{}l@{}}Total energy produced up to time slot $t$ at bus $i$\\  \emph{after} curtailment or load shedding\end{tabular}} \\
                               &                                 \\
\multirow{2}{*}{$f_{ij}[t] = -f_{ji}[t]$} & \multirow{2}{*}{\begin{tabular}[c]{@{}l@{}}Power transferred from bus $i$ to bus $j$ in time slot $t$\\  in the presence of storage \end{tabular}} \\
                               &                                   \\ 
\multirow{2}{*}{$f'_{ij}[t] = -f'_{ji}[t]$} & \multirow{2}{*}{\begin{tabular}[c]{@{}l@{}}Power transferred from bus $i$ to bus $j$ in time slot $t$\\  in the presence of storage\end{tabular}} \\
                               &                                   \\ 
\multirow{2}{*}{$F_{ij}[t] = -F_{ji}[t]$ } & \multirow{2}{*}{\begin{tabular}[c]{@{}l@{}}Total energy transferred from bus $i$ to bus $j$ up to\\time slot $t$ in the presence of storage\end{tabular}} \\
                               &                                                                                                                                              \\
\multirow{2}{*}{$F'_{ij}[t] = -F'_{ji}[t]$} & \multirow{2}{*}{\begin{tabular}[c]{@{}l@{}}Total energy transferred from bus $i$ to bus $j$ up to\\time slot $t$ in the presence of storage\end{tabular}} \\
                               &                                                                                                                                              \\
$p_i[t]$               & Power injected at bus $i$ in slot $t$                                          \\
\multirow{2}{*}{$p^+_i[t]$, $p^-_i[t]$} & \multirow{2}{*}{\begin{tabular}[c]{@{}l@{}}Power produced/consumed at bus $i$ in time slot $t$\\  \emph{before} curtailment or load shedding\end{tabular}} \\
                               &                                   \\ 
\multirow{2}{*}{$\underline{p}^+_i[t]$, $\underline{p}^-_i[t]$} & \multirow{2}{*}{\begin{tabular}[c]{@{}l@{}}Power produced/consumed at bus $i$ in time slot $t$\\  \emph{after} curtailment or load shedding\end{tabular}} \\
                               &                                   \\ 
$S_i$                  & The energy capacity of storage at bus $i$                      \\      \bottomrule                      
\end{tabular}}
\label{symbol}
\end{table}

\section{Problem Formulation} \label{sec:Problem}
\subsection{Storage Model}
A slotted time scheme is introduced to capture the storage SoC behavior in discrete time. Consider a charging/discharging cycle length $T$ divided into $N$ slots: $\mathcal{T} := \{1,...,N\}$, with each of slot length $h=\frac{T}{N}$. This is illustrated in Fig. \ref{timeslot}.
\begin{figure}[H] 
\centering
  \includegraphics[scale=0.7]{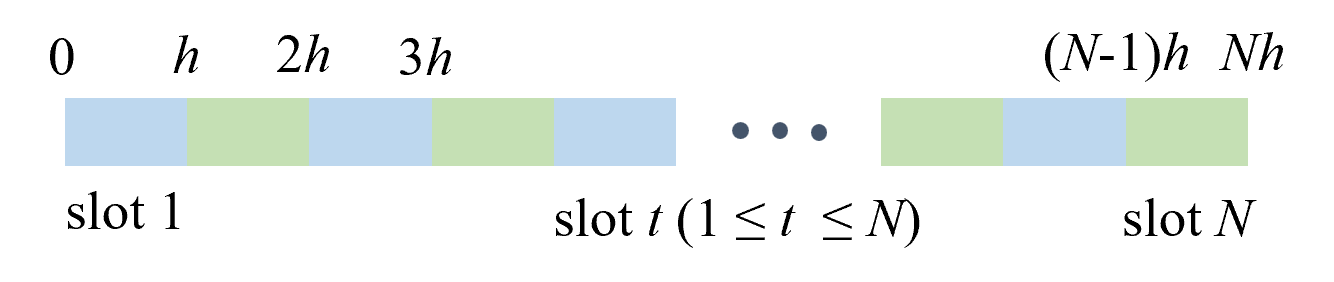}
  \caption{The time slots in a period}
  \label{timeslot}
\end{figure}
At the end of each cycle $\mathcal{T}$, the SoC is required to return to the same value, which is hereafter called the SoC balance assumption in this paper.
\begin{assumption}[SoC Balance]\label{as}
Let $x_{i}[t]\geq0$ represent the SoC of storage installed at bus $i$ after slot $t$, with $x_{i}[0]$ denoting the initial SoC of storage at bus $i$ before cycle $\mathcal{T}$. We require the storage to maintain SoC balance after each charging and discharging cycle $\mathcal{T}$, i.e.,
\begin{equation}\label{as22}
x_{i}[0]=x_{i}[N].
\end{equation} 
\end{assumption}
The selection of period duration $T$ can depend on both the system operator's requirements and the storage technologies. To avoid load-shedding with storage, the choice of $T$ should be large enough so that the total energy generated in $\mathcal{T}$ is greater or equal to the total energy consumed in $\mathcal{T}$. In practice, different storage techniques may use different period lengths $T$ \cite{viswanathan20222022,schmidt2019projecting}. For instance, pumped hydro and compressed air storage typically have a longer duration with less charging/discharging frequency than chemical energy storage. The choice of slot length $h$ should be small enough to make sure the power generated or consumed does not vary much within each time slot.\\
\indent To capture the transmission benefit of storage, only the SoC balance assumption is considered for the storage model. We show that this makes possible a closed-form relationship between storage and transmission line capacity, which provides a capability for quick \emph{back of the envelope} calculations on transmission and storage costs that provide valuable insight. Other more detailed constraints, including storage charging/discharging efficiency and rate, can be included in a more general optimization problem and solved computationally, as we show. 

\subsection{System Description}
Consider a power grid with $\mathcal{B}:=\{1,...,B\}$ buses and $\mathcal{L}:=\{1,...,L\}$ transmission lines. The power generated by bus $i$ in slot $t$ is $p_i[t]$, which could be positive or negative depending on whether the bus is producing energy (e.g.,  a generator or renewable energy source) or consuming energy (like a load). To distinguish these two types of buses, the power produced by $i$ in slot $t$ is defined as $p^+_i[t] := \max\{0,p_i[t]\}$, while the power consumed by $i$ in slot $t$ is $p^-_i[t] := \max\{0,-p_i[t]\}$. Note that $p_i[t] = p^+_i[t] - p^-_i[t]$. \\
\indent To include the generation curtailment and load-shedding scenarios, let $\underline{p}^+_i[t] \in [0, p^+_i[t]]$ be the power produced by bus $i$ to the system after curtailment in slot $t$. Similarly, let $\underline{p}^-_i[t] \in [0, p^-_i[t]]$ be the power consumed by bus $i$ from the system after load shedding in slot $t$.

\subsection{Problem Formulation}
To reduce generation curtailment and load shedding in the future and the accompanying value loss, the system planner needs to analyze multiple scenarios and solutions. In this paper, we are interested in the special relationship between storage and transmission lines in benefiting the system.\\
\indent To get started, let us begin by assuming that generator and load profiles are known before installing new storage and transmission lines. The conventional optimization problem for finding the minimum cost of installing storage and transmission lines is:
\begin{equation}\label{meshedori}
\begin{aligned}
\min &\quad  \sum_{i=1}^B\sum_{j\neq i} \alpha_{ij} c_{ij} + \sum_{i=1}^B \beta_i S_i  \\
+\sum_{i=1}^B \sum_{s=1}^N h & \{ \gamma^+_i(p^+_i[s]-\underline{p}^+_i[s])+\gamma^-_i(p^-_i[s]-\underline{p}^-_i[s])\}, \\
\text { s.t. }   &\forall i\neq j,\quad i,j \in \mathcal{B}, \quad t \in \mathcal{T}:\\
&0 \leq x_{i}[t] \leq S_{i}, \quad 0 \leq x_{i}[0] \leq S_{i},\\
&x_{i}[0]=x_{i}[N],\\
& | f_{ij}[t] | \leq  C_{ij}+c_{ij}, \\
& 0 \leq \underline{p}^+_i[t] \leq p^+_i[t], \quad 0 \leq \underline{p}^-_i[t] \leq p^-_i[t], \\
& C_{ij} \geq 0,\quad c_{ij} \geq 0, 
\end{aligned}   
\end{equation}
where $f_{ij}[t]$ represents the line flow from bus $i$ to bus $j$ in time slot $t$, satisfying $f_{ij}[t] = -f_{ji}[t]$. Above,
\begin{equation} \nonumber
\sum_{i=1}^B \sum_{t=1}^N h \{ \gamma^+_i(p^+_i[t]-\underline{p}^+_i[t])+\gamma^-_i(p^-_i[t]-\underline{p}^-_i[t])\}     
\end{equation}
is the total loss of value during $\mathcal{T}$ due to curtailment and load shedding, while $\sum_{i=1}^B\sum_{j\neq i} \alpha_{ij} c_{ij}$ is the discounted cost of the transmission line, and $\sum_{i=1}^B \beta_i S_i$ is the discounted cost of storage. For the DC power flow model, the optimization problem is linear programming (LP), because $f_{ij}[t]$ can be directly calculated by the power transfer distribution factor (PTDF) matrix $\textbf{H}$, which in matrix form is:
\begin{equation} \label{ptdf}
\boldsymbol{f}[t] =   \textbf{H}\{  (\boldsymbol{p}^+[t]-\underline{\boldsymbol{p}}^+[t]) - (\boldsymbol{p}^-[t]-\underline{\boldsymbol{p}}^-[t]) + (\boldsymbol{x}[t-1] - \boldsymbol{x}[t])\}.
\end{equation}
\indent The decision variables of the problem (\ref{meshedori}) are $\{c_{ij},x_i[0],x_i[t], S_i[t],f_{ij}[t], \underline{p}^+_i[t],\underline{p}^-_i[t]\}$, but after eliminating the equality constraints the decision variables can be reduced to $\{c_{ij},x_i[t], S_i,\underline{p}^+_i[t],\underline{p}^-_i[t]\}$, totaling $L + B + 3BN$ in number.\\
\indent In the planning process, the system operator may choose to look at the historical power profile of each bus as a reference for future planning. Load and generation expansion parameters may also multiply the historical data for a more precise model, but the power profile $p_i[t]$ can also be estimated based on unit commitment (UC) and economic dispatch (ED) before planning new storage and transmission lines.

\section{Main Results} \label{sec:res}
Installing more energy storage or transmission lines can both reduce the value loss of generation curtailment and load shedding. We examine their interrelationships and their fundamental limitations.
\subsection{Cumulative Energy Perspective}
A ``cumulative energy" perspective is introduced to clarify the relationship between storage and transmission lines. 
\begin{definition}[Local Cumulative Energy]\label{d1}
The local cumulative energy $E_i[t]$ is defined as the total energy produced up to time slot $t$ by all generators and loads connected to bus $i$:
\begin{equation} \label{d1e}
E_i[t]:= h\sum_{s=1}^{t} p_{i}[s].
\end{equation}
\end{definition}
\begin{definition}[Local Cumulative Net Energy]\label{d2}
A key role is played by the local cumulative \emph{net} energy $\underline{E}_i[t]$ defined as the total energy produced up to time slot $t$ by all generators and loads connected to bus $i$ \emph{after} curtailment or load shedding:
\begin{equation} \label{d1e}
\underline{E}_i[t]:= h\sum_{s=1}^{t} (\underline{p}^+_{i}[s]-\underline{p}^-_{i}[s]).
\end{equation}
The values of $E_i[t]$ and $\underline{E}_i[t]$ could be positive or negative depending on the bus components.
\end{definition}
\begin{definition}[Transferred Cumulative Energy]\label{d3}
The transferred cumulative energy $F_{ij}[t]$ is defined as the total energy transferred from bus $i$ to bus $j$ until time slot $t$:
\begin{equation}\label{tce}
F_{ij}[t] := h\sum_{s=1}^{t}  f_{ij}[s],
\end{equation}
where $f_{ij}[t]$ is the power transferred from bus $i$ to bus $j$ in time slot $t$ in the presence of storage.
\end{definition}
\indent To quantify the transmission value of energy storage through power flow shaping, the original transferred cumulative energy, in the absence of any additional storage, is introduced for comparison.
\begin{definition}[Original Transferred Cumulative Energy]\label{d4}
The original transferred cumulative energy $F'_{ij}[t]$ is defined as the total energy transferred from bus $i$ to bus $j$ by time slot $t$ \emph{before} considering any new installed storage:
\begin{equation}\label{otce}
F'_{ij}[t] := h\sum_{s=1}^{t}  f'_{ij}[s],
\end{equation}
where $f'_{ij}[t]$ is the power transferred from bus $i$ to bus $j$ in time slot $t$ \emph{before} considering (additional) storage, which can be expressed in matrix form under DC power flow scheme as:
\begin{equation}
\boldsymbol{f'}[t] =   \textbf{H}\{  (\boldsymbol{p}^+[t]-\underline{\boldsymbol{p}}^+[t]) - (\boldsymbol{p}^-[t]-\underline{\boldsymbol{p}}^-[t])\}.
\end{equation}
\end{definition}
Based on the definitions above, the SoC balance assumption (Assumption \ref{as}) can be restated in two other ways.
\begin{corollary}[Equivalent SoC Balance]
The SoC balance assumption in the cumulative energy form is equivalent to either of the equations below:
\begin{equation} \label{eqsoc1}
F_{ij}[N] = F'_{ij}[N] 
\end{equation}
or
\begin{equation} \label{eqsoc2}
     \sum_{j\neq i}F_{ij}[N]=\underline{E}_i[N].
\end{equation}
\end{corollary}
\begin{proof}
After selecting $t=N$ in (\ref{tce}), (\ref{otce}), and (\ref{powerbalance}), we have:
\begin{equation} \label{powerbalanceN}
\begin{aligned}
\boldsymbol{F}[N] -  \boldsymbol{F}'[N]  =  \textbf{H}(\boldsymbol{x}[0] - \boldsymbol{x}[N]), \\  
x_i[N] = x_i[0] + \underline{E}_i[N] - \sum_{j\neq i}F_{ij}[N].    
\end{aligned}
\end{equation}
Then, the SoC balance assumption $x_i[0] = x_i[N]$ also means:
\begin{equation} 
\begin{aligned}
\boldsymbol{F}[N] -  \boldsymbol{F}'[N]= \boldsymbol{0}, \\
\underline{E}_i[N] - \sum_{j\neq i}F_{ij}[N] = 0. 
\end{aligned}
\end{equation}
\end{proof}
Similarly, the minimum storage capacity can be determined by directly analyzing the cumulative energies (\ref{d1e}) and (\ref{tce}).
\begin{theorem}[Minimum Storage Capacity Needed] \label{theoremcap}
The minimum storage capacity $S_i$ to ensure that the power balance at bus $i$ is met is:
\begin{equation} \label{mincap}
\max_{t \in \mathcal{T}}(\underline{E}_i[t] - \sum_{j\neq i}F_{ij}[t]) - \min_{t \in \mathcal{T}}(\underline{E}_i[t] - \sum_{j\neq i}F_{ij}[t]).   
\end{equation}
\end{theorem}
\begin{proof}
To ensure the real-time power balance at any bus $i$, the storage SoC should satisfy
\begin{equation} \label{powerbalance}
x_i[t] = x_i[0] + \underline{E}_i[t] - \sum_{j\neq i}F_{ij}[t] \quad \forall t \in \mathcal{T}.
\end{equation}
\indent Meanwhile, we need $x_i[t] \leq S_i$ for all $t \in \mathcal{T}$, which is equivalent to
\begin{equation} \label{si}
S_i \geq x_i[0] + \underline{E}_i[t] - \sum_{j\neq i}F_{ij}[t] \quad \forall t \in \mathcal{T}.
\end{equation}
\indent Because the SoC should be non-negative, i.e. $x_i[t] \geq 0$, we have
\begin{equation} \label{soc0}
x_i[0] + \underline{E}_i[t] - \sum_{j\neq i}F_{ij}[t]\geq 0 \quad \forall t \in \mathcal{T}.
\end{equation}
\indent In fact, to reduce the storage cost in (\ref{meshedori}), the LP will choose the minimum feasible value of initial SoC:
\begin{equation} \label{soc0min}
x_i[0] = - \min_{t \in \mathcal{T}}(\underline{E}_i[t] - \sum_{j\neq i}F_{ij}[t]).
\end{equation}
By substituting the storage initial SoC (\ref{soc0min}) into (\ref{si}), we have the minimum storage capacity
\begin{equation} \label{soc0min}
S_i = \max_{t \in \mathcal{T}}(\underline{E}_i[t] - \sum_{j\neq i}F_{ij}[t]) - \min_{t \in \mathcal{T}}(\underline{E}_i[t] - \sum_{j\neq i}F_{ij}[t]).
\end{equation}
\end{proof}
Because $x_i[0] = - \min_{t \in \mathcal{T}}(\underline{E}_i[t] - \sum_{j\neq i}F_{ij}[t])$, we can also directly represent the storage SoC from the cumulative energy perspective.
\begin{corollary}[Storage SoC]
The Storage SoC $x_i[t]$ at any time slot $t$ can be expressed as
\begin{equation} \label{minsoc}
\underline{E}_i[t] - \sum_{j\neq i}F_{ij}[t] - \min_{t \in \mathcal{T}}(\underline{E}_i[t] - \sum_{j\neq i}F_{ij}[t]) \quad \forall t \in \mathcal{T}.
\end{equation}
\end{corollary}
Because the line flow $\boldsymbol{f}[t]$ includes the storage SoC $\boldsymbol{x}[t]$ in (\ref{ptdf}), the minimum storage capacity (\ref{mincap}) and the corresponding SoC (\ref{minsoc}) are implicit expressions, but they can be converted to explicit closed-form equations under some conditions.

\subsection{Problem Reformulation}
The cumulative energy analysis bridges the gap between storage and transmission lines, which also allows us to reformulate the original optimization problem (\ref{meshedori}) by switching decision variables from $\{ x_i[t], S_i\}$ to $f_{ij}[t]$. The overall problem after reformulation becomes:\\
\begin{equation}\label{meshedref}
\begin{aligned}
&\min \quad  \sum_{i=1}^B\sum_{j\neq i} \alpha_{ij} c_{ij} + \sum_{i=1}^B \beta_i S_i + \\
 &\sum_{i=1}^B \sum_{s=1}^N h \{ \gamma^+_i(p^+_i[s]-\underline{p}^+_i[s])+\gamma^-_i(p^-_i[s]-\underline{p}^-_i[s])\}, \\
\text { s.t. } &\underline{E}_i[t] = h\sum_{s=1}^{t} (\underline{p}^+_{i}[s]-\underline{p}^-_{i}[s]),\quad F_{ij}[t] = h\sum_{s=1}^{t}  f_{ij}[s],\\
&x_i[0] + \underline{E}_i[t] - \sum_{j\neq i}F_{ij}[t] = 0,\\
&S_i = x_i[0] - \min_{t \in \mathcal{T}}(\underline{E}_i[t] - \sum_{j\neq i}F_{ij}[t]),\\
&\underline{E}_i[N] = \sum_{j\neq i}F_{ij}[N],\\
&| f_{ij}[t] | \leq  C_{ij}+c_{ij}, \\
& 0 \leq \underline{p}^+_i[t] \leq p^+_i[t], \quad 0 \leq \underline{p}^-_i[t] \leq p^-_i[t], \\
& C_{ij} \geq 0,\quad c_{ij} \geq 0, \\
&\forall i\neq j,\quad i,j \in \mathcal{B}, \quad t \in \mathcal{T},
\end{aligned}   
\end{equation}
where the decision variables are converted to $\{c_{ij},x_i[0],S_i,f_{ij}[t], \underline{p}^+_i[t],\underline{p}^-_i[t], \underline{E}_i[t], F_{ij}[t]\}$. \\
\indent Similarly, after eliminating the equality constraints, the decision variables can be simplified to $\{c_{ij},f_{ij}[t],\underline{p}^+_i[t],\underline{p}^-_i[t]\}$. Then, the number of decision variables in the reformulated problem (\ref{meshedref}) is $L+LN+2BN$, which is larger than the original problem (\ref{meshedori}) if $LN-BN-B \geq 0$, and vice versa. Nevertheless, because both problems are LPs, the computation time is expected to be relatively short, provided that the problem size is within reasonable limits.

\subsection{Fundamental Limits}
\indent Two major uncertainties impact the results of optimization problems (\ref{meshedori}) and (\ref{meshedref}). The first stems from the variability in load and generation expansion, which is related to the forecasting inaccuracies of future bus profiles $p^+_i[t]$ and $p_i^-[t]$. Additionally, rapid advancements in materials science, particularly in the domains of battery storage \cite{viswanathan20222022,schmidt2017future} and transmission lines \cite{Chojkiewicz2023Accelerating}, introduce significant unpredictability in the future investment costs for storage ($\beta_i$) and transmission infrastructure ($\alpha_{ij}$).\\
\indent As $\beta_i$ and $\alpha_{ij}$ vary, the solution of (\ref{meshedori}) or (\ref{meshedref}) traces out the optimal trade-off curve of storage vs. transmission lines \cite{boyd2004convex}. Rather than directly comparing the optimal outcomes across various scenarios, we aim to identify the different values of storage and transmission lines and their fundamental limitations.\\
\indent First, SoC balance (Assumption \ref{as}) guarantees that storage neither produces nor consumes net energy over each cycle. If the total energy generated by the system exceeds the total energy consumption over the cycle $\mathcal{T}$, that is, if $\sum_{s\in\mathcal{T}}E_i[s] > 0$, then generation curtailment is unavoidable regardless of the storage or transmission line capacity. Conversely, load shedding is required if $\sum_{s\in\mathcal{T}}E_i[s] < 0$.\\
\begin{assumption}[Energy Balance]\label{eb}
To unveil the deeper relationship between energy storage and transmission lines, the total energy generated by the system is assumed equal to the total energy consumption over the cycle $\mathcal{T}$, i.e.
\begin{equation}
\sum_{i\in\mathcal{B}}E_i[N] = 0.    
\end{equation}
\end{assumption}
Under Assumption \ref{eb}, generation curtailment or load shedding is completely avoidable by installing more storage and transmission lines. Also, the total generation curtailment will be equal to the total load shedding. \\
\indent Typically, the value of the lost load (VOLL), i.e. $\gamma^-_i$, is extremely high. For example, the cost of a one-hour outage for residential load is currently administratively established as \$ 5,122 / MWh in the Electric Reliability Council of Texas (ERCOT) region \cite{voll}. Based on this, the solution of (\ref{meshedref}) can be approximated by assuming no load shedding during the cycle $\mathcal{T}$, i.e.
\begin{equation}
p^+_i[t]=\underline{p}^+_i[t],\quad p^-_i[t]=\underline{p}^-_i[t], \quad \forall t \in \mathcal{T},
\end{equation}
or
\begin{equation}
E_i[t] = \underline{E}_i[t], \quad \forall t \in \mathcal{T},
\end{equation}
which allows the original power flow or cumulative transferred energy in Definition \ref{d4} to be directly calculated \emph{before} solving the optimization problem. The corresponding simplified linear program is:
\begin{equation}\label{meshedsim}
\begin{aligned}
\min \quad  &\sum_{i=1}^B\sum_{j\neq i} \alpha_{ij} c_{ij} + \sum_{i=1}^B \beta_i S_i \\
\text { s.t. } & \quad p^+_i[t]=\underline{p}^+_i[t],\quad p^-_i[t]=\underline{p}^-_i[t], \\
& \quad \text{and} \quad \text{all constraints in }(\ref{meshedref}). 
\end{aligned}   
\end{equation}
\subsubsection{Fundamental Limit of Transmission Lines} \quad \\
\indent Suppose future technology makes the cost of constructing or improving the capacity of transmission lines extremely cheap, or the system already has enough transmission capability. However, the system might \emph{still} need storage to avoid load shedding in (\ref{meshedsim}), due to a \emph{fundamental limit of transmission lines}.\\
\indent Since the transmission lines only serve as a \emph{spatial} energy transfer mechanism, storage is necessary whenever there is a \emph{temporal} imbalance in the total system power, irrespective of the transmission line capacity available. If the system has infinite transmission capability, the storage can be placed anywhere without limitations. Under this situation, the total minimum storage capacity needed is given explicitly in Theorem \ref{totalmin}.
\begin{theorem} [Total Minimum Storage Capacity Needed] \label{totalmin}
Under Assumption \ref{eb} and with sufficient transmission line capacity, if the system wants to avoid load shedding, the minimum total storage capacity needed is
\begin{equation} \label{mincaptotal}
\max_{t \in \mathcal{T}}(\sum_{i=1}^B E_i[t]) - \min_{t \in \mathcal{T}}(\sum_{i=1}^B E_i[t]).   
\end{equation}    
\end{theorem}
\begin{proof}
After summing up (\ref{si}) over all buses, the total storage capacity satisfies
\begin{equation}
\sum_{i=1}^B S_i \geq \sum_{i=1}^Bx_i[0] + \sum_{i=1}^B\underline{E}_i[t] - \sum_{i=1}^B\sum_{j\neq i}F_{ij}[t] \quad \forall t \in \mathcal{T}.    
\end{equation}
Noting that $E_i[t] = \underline{E}_i[t]$ and $F_{ij}[t] = - F_{ji}[t]$, we have:
\begin{equation} \label{sumscap}
\sum_{i=1}^B S_i \geq \sum_{i=1}^Bx_i[0] + \sum_{i=1}^BE_i[t] \quad \forall t \in \mathcal{T}.    
\end{equation}
\indent Similarly, after summing up (\ref{soc0}) over all buses, the total storage initial energy satisfies
\begin{equation} \label{sums0}
\sum_{i=1}^B x_i[0]  \geq -\sum_{i=1}^BE_i[t] \quad \forall t \in \mathcal{T}.    
\end{equation}
After substituting (\ref{sums0}) into (\ref{sumscap}), we have:
\begin{equation} 
\sum_{i=1}^B S_i \geq \max_{t \in \mathcal{T}}(\sum_{i=1}^B E_i[t]) - \min_{t \in \mathcal{T}}(\sum_{i=1}^B E_i[t]),   
\end{equation}
with equality being achieved when the system possesses sufficient transmission line capacity.
\end{proof}
\subsubsection{Fundamental Limit of Energy Storage} \quad \\
\indent On the other hand, the fundamental limit of energy storage, meaning the minimum required transmission line capacity to avoid load shedding, is less obvious. The following bound holds regardless of the storage capacity even if the latter is infinite.
\begin{theorem}[Minimum Line Capacity Needed] \label{linecaplim}
Under Assumption \ref{eb} and sufficient storage capacity, if the system is to avoid load shedding, then the minimum capacity of transmission line $(i,j)$ is the mean value of the original power flow over the cycle $\mathcal{T}$, i.e.,
\begin{equation} \label{meancap}
\frac{1}{N} \sum^{N}_{s=1}|f'_{ij}[s]|.
\end{equation}
\end{theorem}
\begin{proof}
This property is easily proved from the cumulative energy perspective. First, based on Definition \ref{d3}, we have
\begin{equation}
 |F_{ij}[N]| =  h\sum_{s=1}^{N}  |f_{ij}[N]| \leq  h\sum_{s=1}^{N} (C_{ij}+c_{ij}) =  hN(C_{ij}+c_{ij}). 
\end{equation}
On the other hand, the equivalent SoC assumption (\ref{eqsoc1}) can be expressed as
\begin{equation}
 F_{ij}[N] = F'_{ij}[N] =  h\sum_{s=1}^{N}  f'_{ij}[N].
\end{equation}
Then we have
\begin{equation}
h\sum_{s=1}^{N}  |f'_{ij}[N]| \leq  hN(C_{ij}+c_{ij}), 
\end{equation}
which also means
\begin{equation}
C_{ij}+c_{ij} \geq  \frac{1}{N} \sum^{N}_{s=1}|f'_{ij}[s]|.
\end{equation}
\end{proof}
Theorem \ref{linecaplim} not only points out the fundamental limits of energy storage but also allows the decision maker to briefly assess the minimum transmission capacity needed for the future \emph{without} solving any optimization problem. Furthermore, under the minimum line capacity, the corresponding storage capacity and SoC have \emph{explicit} closed-form and expressions.
\begin{corollary}[Closed-form Expression of Storage] \label{closedform}
When the transmission line capacity is equal to the minimum line capacity needed, i.e. $C_{ij}+c_{ij} =  \frac{1}{N} \sum^{N}_{s=1}|f'_{ij}[s]|$, the minimum storage capacity needed is given by the following closed-form expression
\begin{equation} \label{closedcap}
S_i = \max_{t \in \mathcal{T}}(E_i[t] - t\sum_{j\neq i}f'_{ij}[t]) - \min_{t \in \mathcal{T}}(E_i[t] - t\sum_{j\neq i}f'_{ij}[t]).  
\end{equation}
Similarly, the corresponding SoC of transmission-asset storage is
\begin{equation} \label{closedsoc}
E_i[t] - t\sum_{j\neq i}f'_{ij}[t] - \min_{t \in \mathcal{T}}(E_i[t] - t\sum_{j\neq i}f'_{ij}[t]), \quad \forall t \in \mathcal{T}.
\end{equation}
\end{corollary}

\section{Discussion} \label{sec:dis}
\subsection{Scenario Generation}
\indent Both the conventional problem (\ref{meshedori}) and its reformulation (\ref{meshedref}) just reflect the optimal solution over one cycle $\mathcal{T}$. The investment in transmission lines and storage might be a long-term process. To make the decision suitable for future uncertainty, the final results typically need to be based on multiple likely scenarios. This can be done by extending the deterministic problem to many data-driven uncertainty programming approaches, including stochastic, robust, and chance-constrained optimization \cite{roald2023power,zhang2024efficient}. \\
\indent The generation of future nodal energy profile scenarios typically involves two steps. Initially, system planners estimate the expansion of load and renewable energy generation using historical data, weather forecasts, and other information such as electric vehicle penetration. This process may yield various scenarios, including those characterized by low or high wind conditions. Subsequently, to simulate the operational dynamics of the power grid, UC and ED models are employed to manage the output of dispatchable generators.

\subsection{Contingency Analysis}
Before considering contingencies, after integrating enough transmission asset storage into grids, the transmission planning paradigm shifts from designing for the \emph{peak} line flow to designing for the \emph{average} line flow of each cycle. Because transmission planning is a long-term process faced with high uncertainty, it is impossible to take all contingencies into account. Thus it is largely based on current grid status and engineer-defined criteria \cite{choi2007transmission}. To reduce the problem's complexity, contingency analysis is typically applied only to some serious scenarios. Supposing the operator has generated the load and renewable forecasts for future years, the question is how to select the serious scenarios.\\
\indent Assuming the storage charging cycle is one day, based on Theorem \ref{linecaplim}, the minimum line capacity needed is the original mean line flow value of each day. To select the serious scenarios (days), we first screen the minimum capacity of all the lines over the future years of interest and pick the maximum value for each line. Then, every line is associated with one specific day when it reaches the maximum mean line flow, which is then defined as a ``serious scenario" of the system. Because most of the line's maximum mean line flow happens on the same day, the total number of serious scenarios is far less than the total number of lines. This is illustrated in the case studies below.

\section{Case Studies} \label{sec:case}
Three test systems are implemented in this section for different purposes. First, a 2-bus toy textbook example is introduced to help readers obtain an intuitive understanding of the relationship between storage and transmission assets, especially the fundamental limits of both storage and transmission lines.\\
\indent In large-scale power systems, technology iterations and public concerns contribute to significant challenges in estimating parameters for models (\ref{meshedori}) or (\ref{meshedref}). While deregulation in the electricity market facilitates the system's autonomous approximation of the optimal point, the fundamental limits discussed in this paper, offer valuable insights for market mechanism designers and policymakers. \\
\indent To better unveil the transmission value of storage, the real-time power balance for the large system is assumed to be maintained by dispatchable generators. The RTS 24-bus system \cite{ordoudis2016updated} is used to further compare the minimum storage capacity under \emph{different} renewable output levels, where the efficiency of the reformulated problem (\ref{meshedref}) is also compared with the conventional problem (\ref{meshedori}). The more complicated RTS-GMLC system \cite{barrows2019ieee} and the Texas synthetic grid \cite{birchfield2016grid}, where the renewable and load profiles have been estimated for the \emph{whole} year, allow further contingency analysis based on serious scenarios.\\
\indent The charging and discharging cycle of storage is all assumed to be a day (24 hours) with a 4-hour duration\footnote{The storage duration means the amount of time storage can discharge at its power capacity before depleting its energy capacity.}. The basic simulation step is set to one hour, resulting in 24 simulation steps per day. After assuming or computing the bus profiles, all the problems are solved using 64 GB RAM on the Intel XEON-10885M CPU (2.4GHz). The mathematical models were formulated using YALMIP on Matlab R2023a and solved using Gurobi v9.5.
\subsection{2-Bus Example}
The 2-bus toy textbook example represents power transmission from a generation site (Bus 1) to a demand side (Bus 2). In this example, the generators, located in rural areas, are assumed to be renewable energy sources, whereas the demand represents a load center situated far from the generation site.
\begin{figure}[H]
\centering
  \includegraphics[scale=0.4]{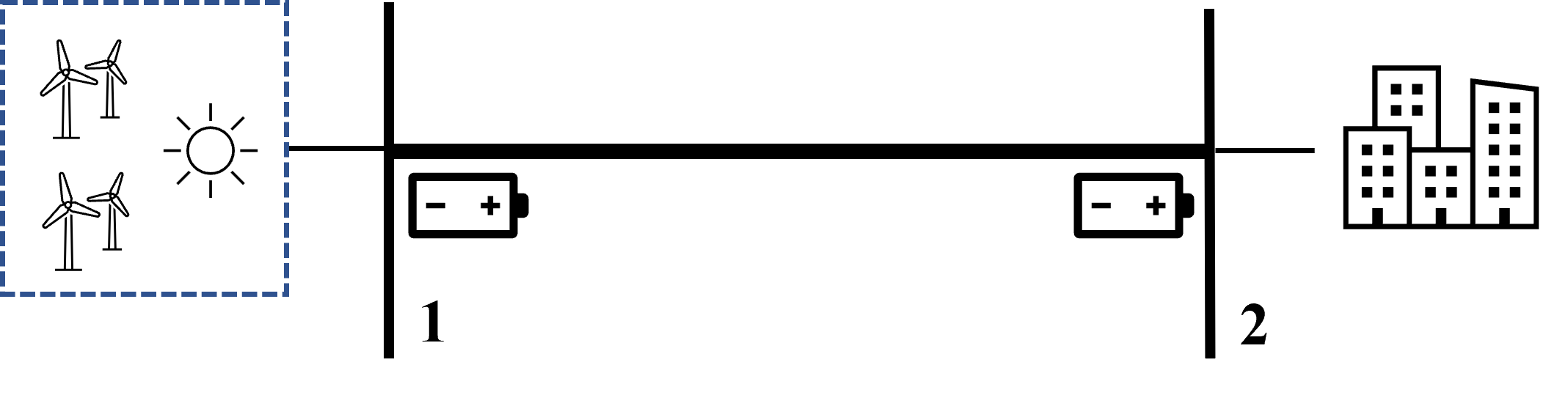}
  \caption{The 2-bus system}
  \label{2bus2}
\end{figure}
The system demand curve is the total load on the peak demand day of ERCOT in 2023 (August 10th) \cite{ercot}, while the supply curve is the total renewable generation output on the same day. To meet the energy balance assumption (Assumption \ref{eb}), both the demand and supply curves are proportionally adjusted to the same average value, i.e. 100 MW. Because the system is not guaranteed to have a real-time power balance, energy storage and a transmission line are both needed to avoid load shedding. This assumption is true for systems with high-penetration renewable resources and limited dispatchable generators. Fig. \ref{2bus2} illustrates this situation where the supply is solely from renewable resources.\\
\indent To unveil the substitutability between storage and transmission lines, the optimization problem (\ref{meshedsim}) is solved multiple times by tuning the cost parameters $\alpha_{12}$, $\beta_1$ and $\beta_2$. Assuming the existing transmission line capacity $C_{12} = 0$, the optimal storage and transmission line capacity under different cost parameters are shown as the three-dimensional surface in Fig. \ref{3D}, where the color of the surface represents the value of the cost ratio, defined as the mean cost of storage divided by the cost of the transmission line, i.e. $\frac{\beta_1+\beta_2}{2\alpha_{12}}$.
\begin{figure}[H]
\centering
  \includegraphics[scale=0.46]{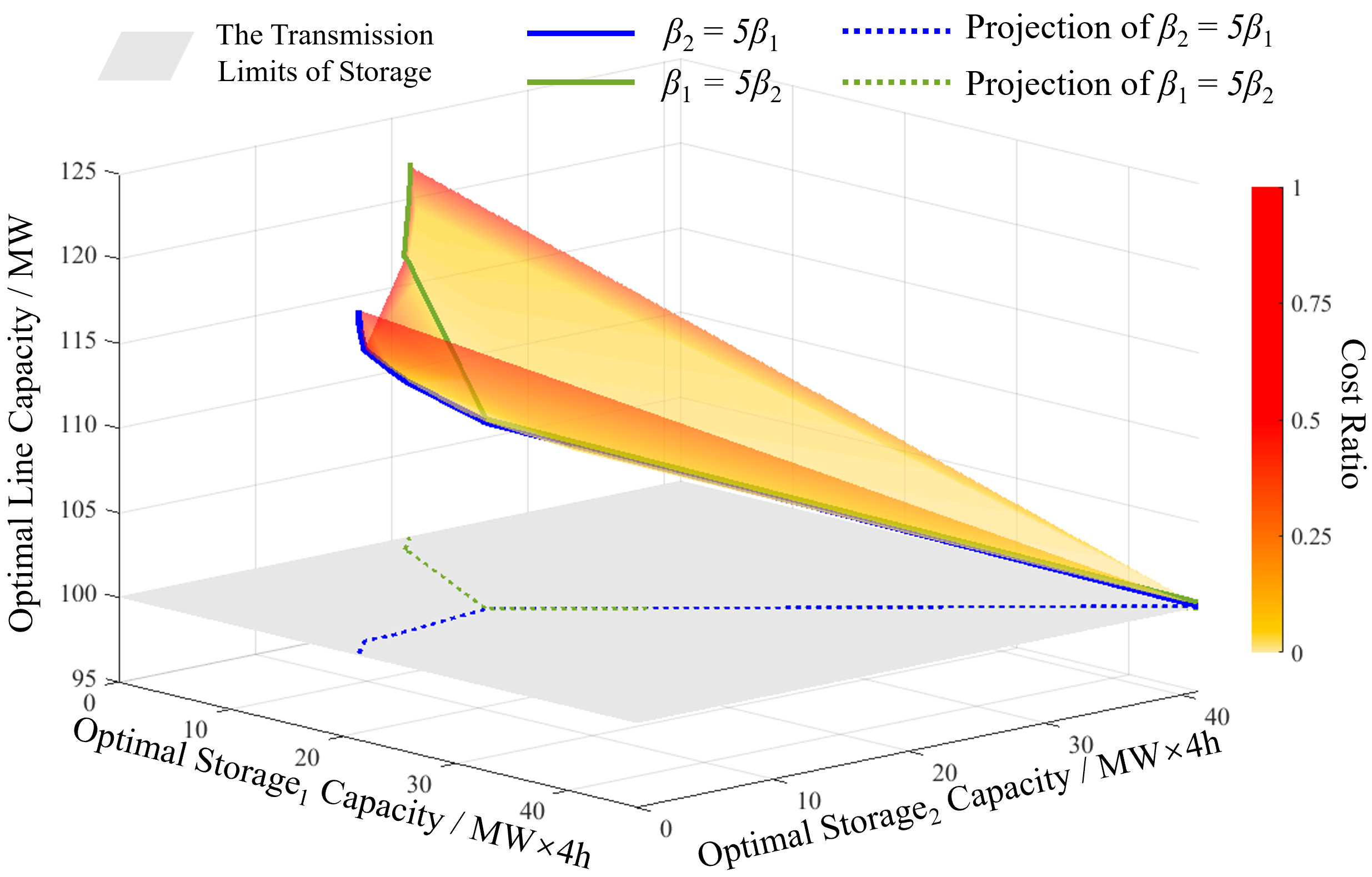}
  \caption{The optimal storage and transmission line capacity under different cost parameters}
  \label{3D}
\end{figure}
\indent The solutions reflect that there exists a minimum line capacity needed or the transmission limits of storage, which is illustrated as the grey plane. Furthermore, considering the cost of storage at different locations might be different, the two special cost cases $\beta_2 = 5\beta_1$ and $\beta_1 = 5\beta_2$ are depicted by blue and green curves respectively. 
\begin{figure}[H]
\centering
  \includegraphics[scale=0.45]{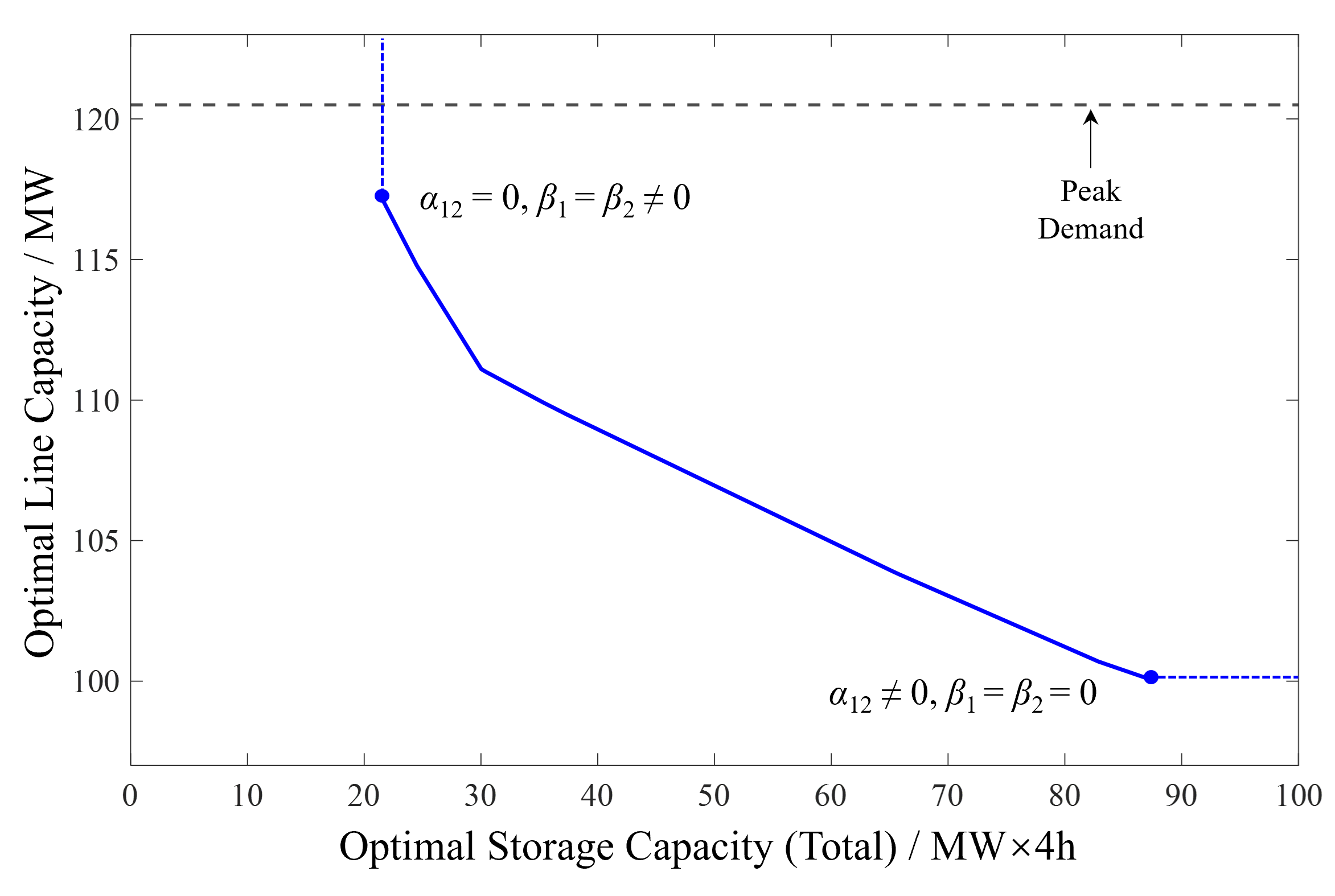}
  \caption{The total optimal storage and transmission line capacity under different cost parameters}
  \label{non-storage}
\end{figure}
\indent For better visualization of total storage capacity, we assume the cost of storage is equal in both buses, i.e. $\beta_1 = \beta_2$. The optimal solution is decided by the relative values of $\alpha_{ij}$ and $\beta_i$, i.e. $\frac{\alpha_{12}}{\beta_1}$ or $\frac{\alpha_{12}}{\beta_2}$. If the cost of energy storage is fixed, then increasing transmission line cost means installing more storage is a better solution, and vice versa. Fig. \ref{non-storage} illustrates the trade-off between storage and transmission lines under different cost situations.
\begin{figure}[H]
\centering
  \includegraphics[scale=0.46]{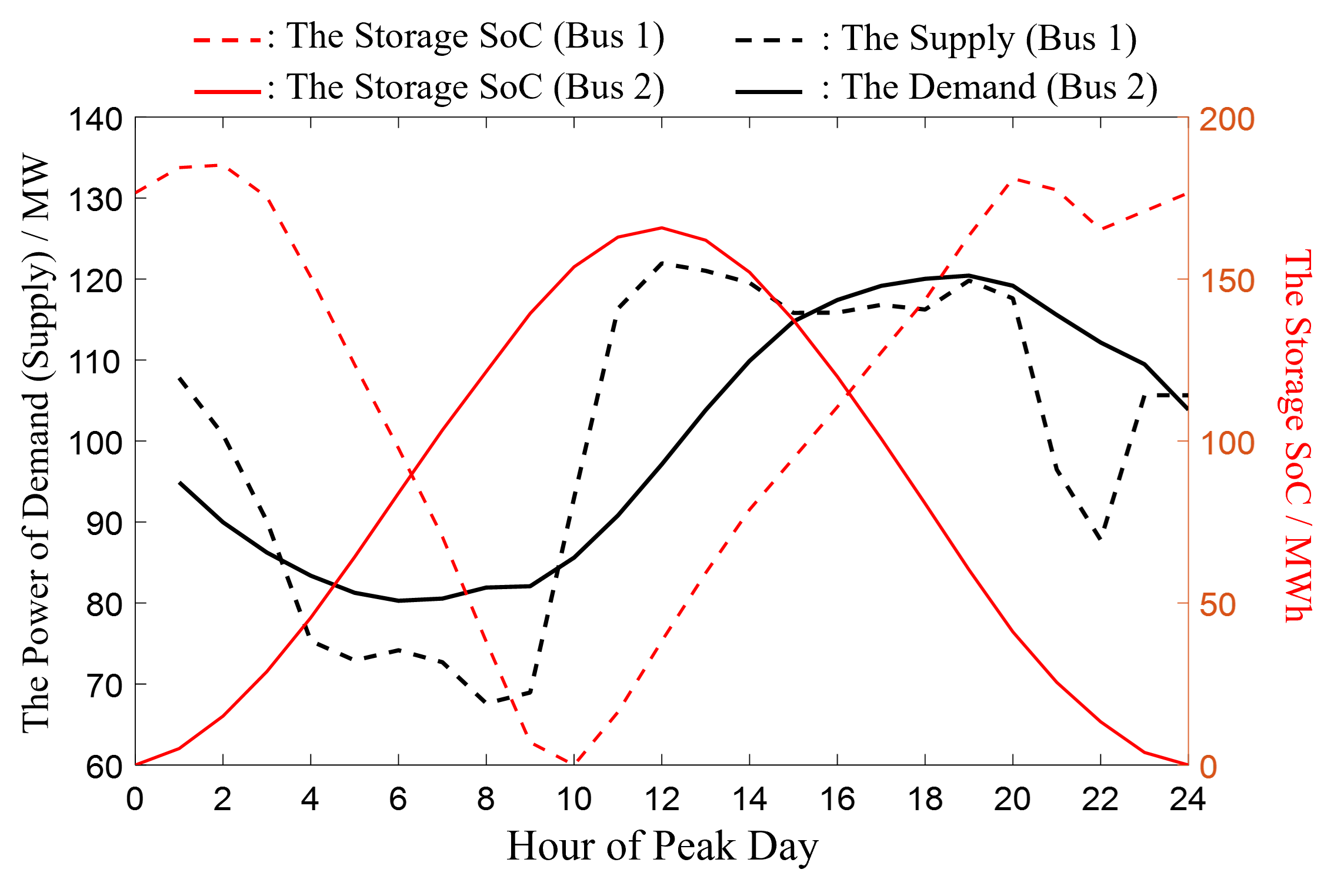}
  \caption{The storage SoC behavior under minimum transmission line capacity (2-bus system)}
  \label{non-soc}
\end{figure}
\indent It is clear that there exist some fundamental limitations for both storage and transmission lines. For the storage, even if the investment cost $\beta_1 = \beta_2 = 0$ and $\alpha_{12} \neq 0$, the system still cannot avoid load shedding or renewable curtailment solely relying on storage. The corresponding minimum transmission line capacity directly calculated from (\ref{mincap}) in Theorem \ref{linecaplim} is the same as obtained by the optimization program in Fig. \ref{non-storage}, that are both equal to the mean value of the load curve (100 MW). Furthermore, the storage capacity and SoC behavior also have explicit closed-forms (\ref{closedcap}) and (\ref{closedsoc}), which are shown in Fig. \ref{non-soc}.\\
\indent Similarly, if $\alpha_{12} = 0$ and $\beta_1 = \beta_2 \neq 0$, the system still needs minimum storage capacity to ensure a real-time power balance between supply and demand. It has the explicit closed-form (\ref{mincaptotal}) resulting in the same value as obtained by solving the optimization program, i.e. 21.4 MW$\cdot$4h. 

\subsection{Modified RTS 24-bus System}
The 24-bus System considered next is modified from the IEEE Reliability Test System (RTS-24) by the addition of six wind farms and one solar farm to model high renewable energy penetration \cite{ordoudis2016updated}. The wind generators are situated at buses 3, 5, 7, 16, 21, and 23, each having a capacity of 150MW, while the solar farm is connected to bus 22 with a capacity of 200MW (see Fig. 8 in \cite{zhang2024efficient}). Both wind and solar profiles are scaled from the total wind and solar output of ERCOT \cite{ercot}, and more renewable scenarios can also be generated based on the method mentioned in \cite{zhang2024efficient}.\\
\indent The modified RTS 24-bus is primarily used to compare different problem formulations under different renewable penetration scenarios. We select the dates 5.29.2022 (high wind) and 10.31.2022 (low wind) as two typical renewable generator output profiles, with averages of 69.2\% and 13.7\% wind output of the total wind capacity respectively. \\
\indent After solving the economic dispatch problem based on the \emph{same} load profile given in \cite{ordoudis2016updated} with different wind profiles, the total minimum line capacity barely varies: 4769.5 MW on high wind days and 4775.0 MW on low wind days. Fig. \ref{ori_cap} compares the original peak line flow value and the minimum line capacity needed on different days, where the difference represents the potential transmission ``value" of energy storage. 
\begin{figure}[H]
\centering
  \includegraphics[scale=0.4]{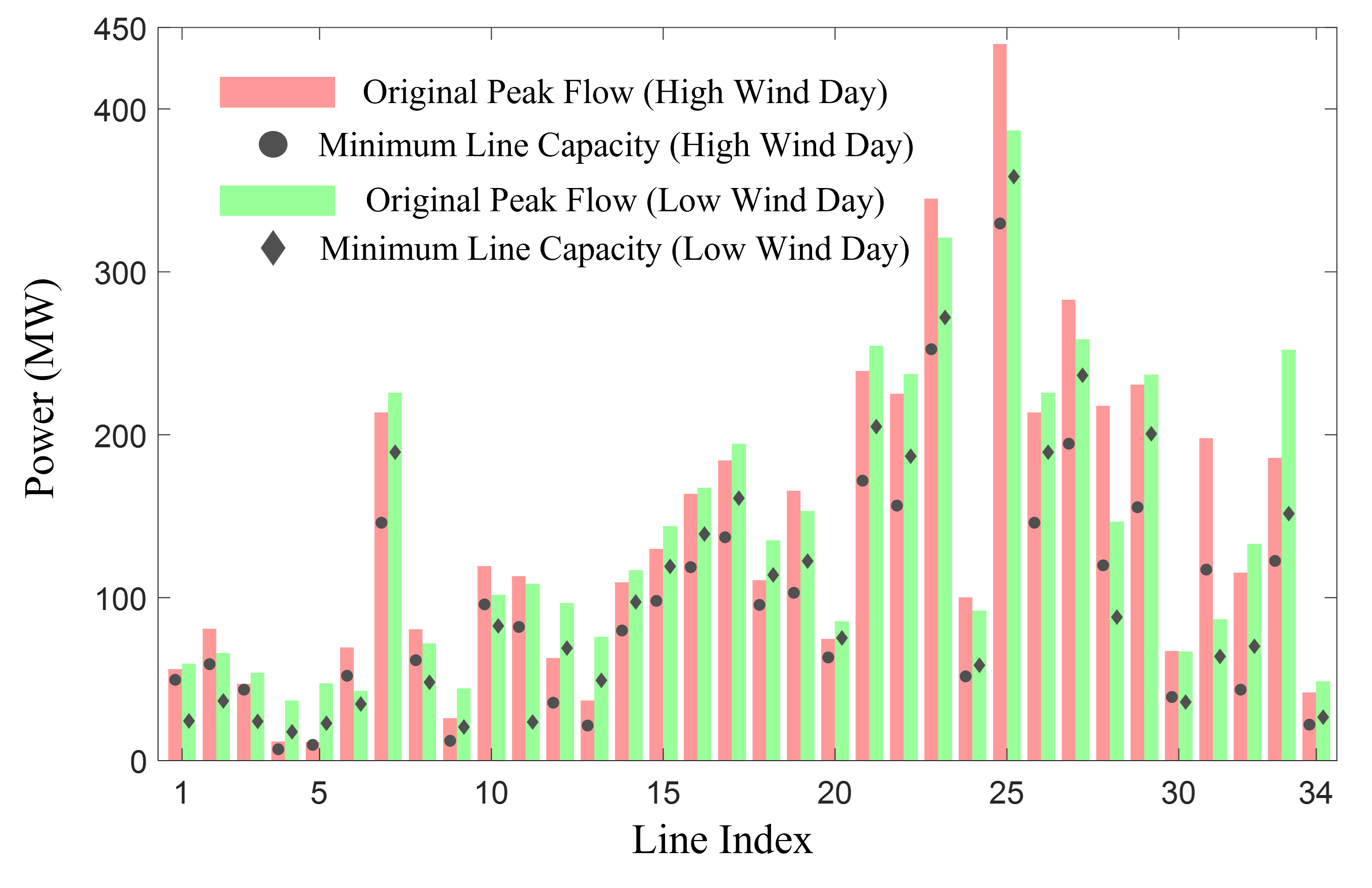} 
  \caption{The original peak line flow value and the minimum line capacity needed on different days}
  \label{ori_cap}
\end{figure}
\indent  If the operator wants to reduce the original peak line flow value by 5\% for each line, which amounts to a 239 MW total line capacity reduction, the minimum total storage capacity needed can be solved for either through the conventional optimization problem (\ref{meshedori}) or the reformulated problem (\ref{meshedref}) by replacing the objective function with $\sum_{i=1}^B S_i$. The comparison of these two methods is shown in TABLE \ref{t1}.
\begin{table}[H]
\caption{The minimal total storage capacity needed to reduce the peak line flow by 5\% based on different formulations}
\centering
\begin{tabular}{cccc}
\toprule
\textbf{Day}                        & \textbf{Method}      & \textbf{Storage Capacity} & \textbf{Solving Time} \\ \midrule
\multirow{2}{*}{\textbf{High Wind}} & Conventional (\ref{meshedori}) & 115.8 MW$\cdot$4h      & 0.26 s  \\
                                    & Reformed (\ref{meshedref}) & 115.8 MW$\cdot$4h      & 0.50 s   \\
\multirow{2}{*}{\textbf{Low Wind}}  & Conventional (\ref{meshedori}) & 100.2 MW$\cdot$4h      & 0.36 s  \\
                                    & Reformed (\ref{meshedref}) & 100.2 MW$\cdot$4h      & 0.48 s   \\\bottomrule
\label{t1}
\end{tabular}
\end{table}
\indent The minimal storage capacity is the same value though given by two different formulations (\ref{meshedori}) and (\ref{meshedref}), but the optimal solution for storage SoC at each time $t$ is not unique. The storage achieves the transmission value by acting as a \textit{flow response} technique to reshape the line flow profile. Because the line flow is the decision variable in (\ref{meshedref}), it is also helpful to quantify the transmission capability of storage for future market design problems.
\subsection{RTS-GMLC 73-bus System}
The RTS-GMLC 73 bus system incorporates contemporary generating resources equipped with distinct heat rate assignments derived from real-world data. Additionally, it integrates a geographically coherent 5-minute time series dataset for wind, solar, and load in the whole year 2020 \cite{barrows2019ieee}. All the data is downloaded from \cite{gmlc} without any changes, except that the DC transmission line in the original system is not considered in this paper.\\
\indent The energy profiles of each bus are solved by the unit commitment problem with a 36-hour look ahead window and ED with 15-minute resolution in advance with infinite transmission line limits for the whole year based on Prescient \cite{knueven2020novel}, an open-source modeling package written in Python for the large-scale power system UC and ED.
\begin{table}[H]
\caption{The total peak line flow value and minimum line capacity needed (\ref{meancap}) - RTS GMLC system}
\centering
\begin{tabular}{ccc}
\toprule
\textbf{Peak Line Flow} & \textbf{Minimum Line Capacity}  & \textbf{Reduced Percentage} \\ \midrule
30686.1 MW                     & 19401.7 MW                & 36.8 \%   \\ \bottomrule                           \label{year_peak_mean}                                
\end{tabular}
\end{table}
\indent Instead of only focusing on two specific days in Fig. \ref{ori_cap}, the minimum line capacity needed for the RTS GMLC system is screened throughout the whole year. After comparing the daily power flow value of each line, the maximum line flow value during the whole year is defined as the original peak line flow value, while the minimum line capacity needed is the maximum daily \emph{mean} line flow value over the year. TABLE \ref{year_peak_mean} compares these two results, which shows the benefit of designing transmission lines for the \emph{average} line flow (with sufficient storage) rather than for the \emph{peak} line flow.\\
\indent To select the serious scenarios for contingency analysis, each transmission line is associated with one serious day when its daily mean line flow value is maximum over the year. The system has 120 transmission lines, but the total number of ``serious" days is only 43 days, which illustrates the concentration property of serious scenarios. For example, 9 transmission lines reach their maximum daily mean line flow spontaneously on the same day (August 12th). The frequency of serious days over the year is shown in Fig. \ref{seriousday}.
\begin{figure}[H]
\centering
  \includegraphics[scale=0.4]{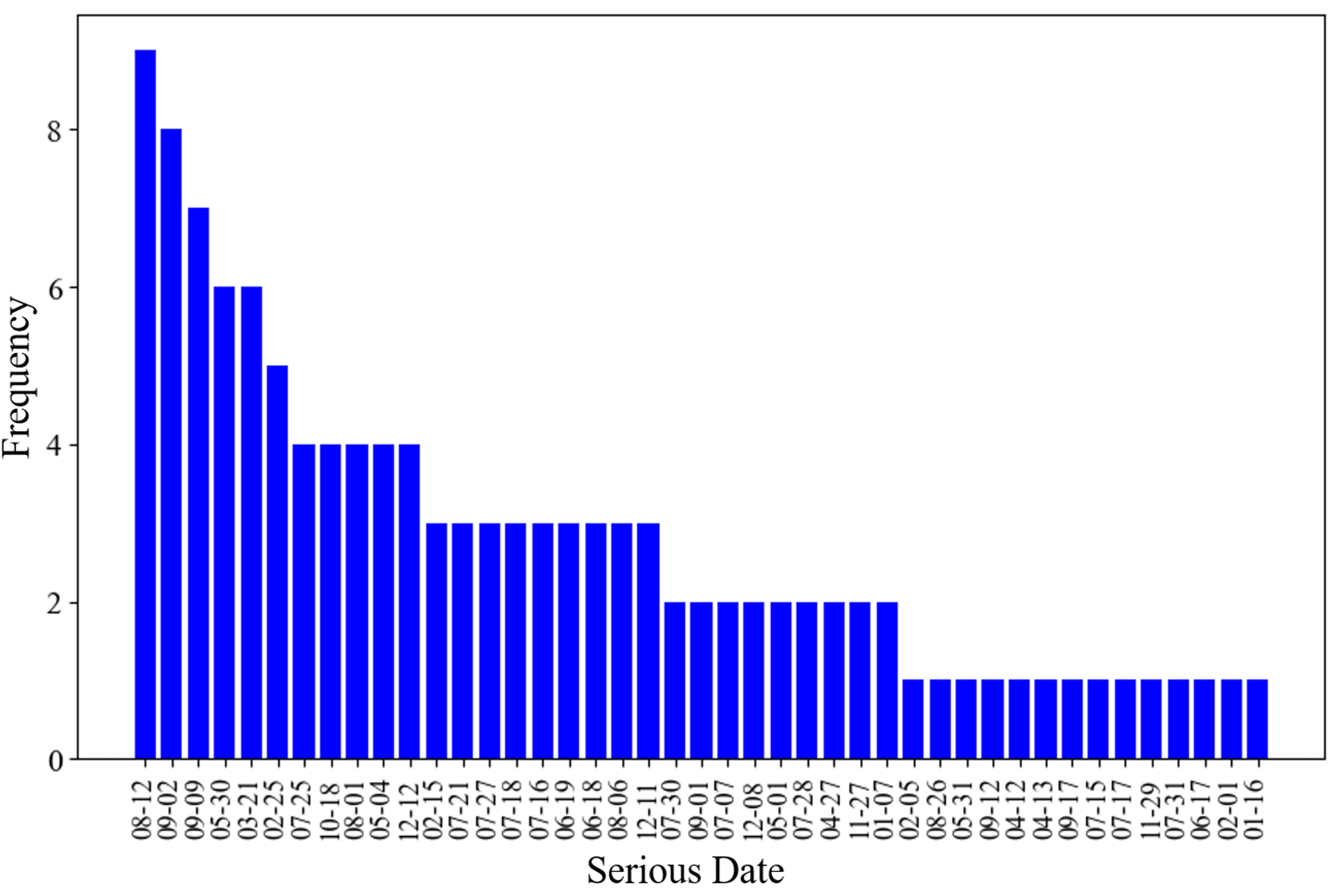} 
  \caption{The frequency of serious days- RTS GMLC system}
  \label{seriousday}
\end{figure}
\indent Because the system line flow is constant under sufficient transmission asset storage, the contingency analysis for each day is employed for the mean flow scenario. The $N$-1 transmission line contingency analysis is applied to the above 43 serious days, where both line and element tripping contingency are considered. Under the line tripping contingencies, the line capacity limit should meet security requirements for all possible $N$-1 transmission constraint situations. For the element tripping contingency, all the elements connected with the same bus are assumed to be cut off together, while the consequent power imbalance is covered proportionally by the elements at other buses.\\
\indent The minimum line capacity needed increases after considering the contingency scenarios, but the total capacity is \emph{still} less than the total peak line flow value, which is the line capacity needed without considering storage and contingency. If comparing the results with the given line capacity in \cite{gmlc}, about 50 \% of total line capacity is reduced, as shown in TABLE \ref{contingencycap}. 
\begin{table}[H]
\caption{The (total) fundamental line capacity limits considering $N$-1 contingencies}
\centering
\begin{tabular}{cccc}
\toprule
\multirow{2}{*}{\textbf{\begin{tabular}[c]{@{}c@{}}Contingency\\  Type\end{tabular}}}
&\multirow{2}{*}{\textbf{\begin{tabular}[c]{@{}c@{}}Minimum\\  Line Capacity\end{tabular}}} 
&\multirow{2}{*}{\begin{tabular}[c]{@{}c@{}}\textbf{Percentage Reduction From} \\Current Installed Capacity\tablefootnote{ Assumed to be the given parameter in \cite{gmlc}.}\end{tabular}} \\ \\ \midrule
\textbf{Line Trip}                & 24937.2 MW                 &   46.6\% \\ 
\textbf{Element Trip}             & 23099.4 MW                &   50.5\% \\ 
\textbf{Both}                     & 26389.2 MW               &   43.5\% \\  \bottomrule                           
 \label{contingencycap}                                
\end{tabular}
\vspace{-5pt}
\end{table}
\subsection{Synthetic Texas Grid}
We show the potential transmission value of energy storage on the latest 6716-bus synthetic grid on a footprint of Texas \cite{birchfield2016grid}. This grid consists of 731 generation units, with a portfolio of 475 gas, 23  coal, 4 nuclear, 22 hydro, 153 wind, and 36 utility-scale solar power plants. The forecast and actual solar, wind, and load data are produced by the U.S. Department of Energy (DOE) PERFORM Program, where we extract the whole year data from 09/01/2022 to 08/31/2023 for analysis \cite{performance}. The ED time interval is set to one hour after 36 hours look ahead UC to save computing time. The whole problem is also solved under infinite transmission line limits based on Prescient.
\begin{table}[H]
\caption{The total peak line flow value and minimum line capacity needed (\ref{meancap}) - Texas synthetic grid}
\centering
\begin{tabular}{ccc}
\toprule
\textbf{Peak Line Flow} & \textbf{Minimum Line Capacity}  & \textbf{Reduced Percentage} \\ \midrule
11.59 $\times 10^5$  MW                     & 9.01 $\times 10^5$ MW                & 22.3 \%   \\ \bottomrule                           \label{year_peak_mean2}                                
\end{tabular}
\vspace{-5pt}
\end{table}
\indent Similarly, after comparing the daily power flow value of each line, TABLE \ref{year_peak_mean2} compares the total peak line flow of the system and the total minimum line capacity. Without considering contingency, designing transmission lines for the \emph{average} line flow (with sufficient storage) rather than for the \emph{peak} line flow will bring a 22.3 \% reduction capacity based on this previous year's scenarios. \\
\indent The Texas synthetic grid, which comprises 9,140 transmission lines, experiences a similar concentration property where 80\% of the lines' maximum daily mean flow value occurs within just 22 days of the year. For instance, 1,287 transmission lines reach their maximum daily mean line flow spontaneously on the same day (December 23rd), and the frequency of these 22 serious days is shown in Fig. \ref{seriousday2}. 
\begin{figure}[H]
\centering
  \includegraphics[scale=0.4]{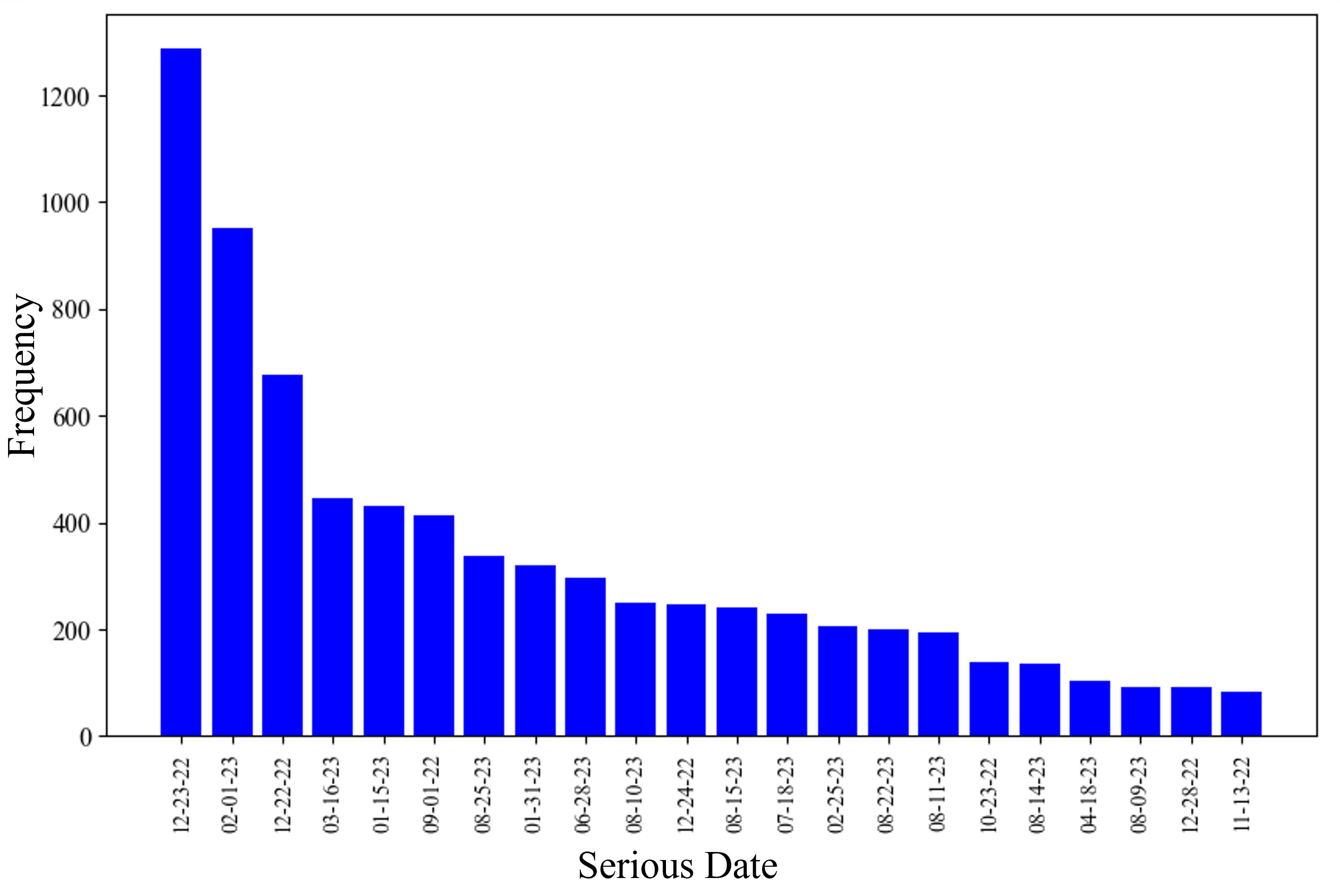} 
  \caption{The frequency of serious days - Texas synthetic grid}
  \label{seriousday2}
\end{figure}
\indent Considering the system size, the $N$-1 transmission line contingency analysis is only applied to the above 22 serious days, where both line and element tripping contingency are considered. The contingency analysis model is the same as with the RTS GMLC system. \\ 
\indent The current design of line capacity is intended to guarantee system safety during peak scenarios under contingency conditions. Consequently, compared with the line capacity data given in \cite{performance}, the transmission potential of storage is estimated to enable the reduction in line capacity requirements under various conditions: 46.3 \% when only considering line tripping, 62.8 \% when only accounting for element tripping, and 46.2 \% when both line and element tripping are taken into account. The detailed results are shown in TABLE \ref{contingencycap}. 
\begin{table}[H]
\caption{The (total) fundamental line capacity limits considering $N$-1 contingencies}
\centering
\begin{tabular}{ccc}
\toprule
\multirow{2}{*}{\textbf{\begin{tabular}[c]{@{}c@{}}Contingency\\  Type\end{tabular}}}
&\multirow{2}{*}{\textbf{\begin{tabular}[c]{@{}c@{}}Minimum\\  Line Capacity\end{tabular}}} 
&\multirow{2}{*}{\begin{tabular}[c]{@{}c@{}}\textbf{Percentage Reduction From} \\Current Installed Capacity\tablefootnote{Assumed to be the given parameter in \cite{performance}.} \end{tabular}} \\ \\ \midrule
\textbf{Line Trip}                & 13.99 $\times 10^5$ MW                 &  46.3 \% \\ 
\textbf{Element Trip}             & 9.69 $\times 10^5$ MW                  &   62.8\% \\ 
\textbf{Both}     & 14.01 $\times 10^5$ MW                 &   46.2\% \\  \bottomrule                           
 \label{contingencycap}                                
\end{tabular}
\vspace{-5pt}
\end{table}
\section{Conclusion} \label{sec:con}
This paper addresses the transmission value of energy storage by developing a cumulative energy-based model that elucidates the interconnection between storage and line flow. The fundamental limitations of both storage and transmission lines are revealed in an explicit closed-form way, which can be used to determine the minimum transmission capacity and the minimum storage capacity essential for the grid operation. The potential of storage to reduce required line capacity is valid in a wide range of systems with and without considering contingencies. \\
\indent To illuminate the universal relationship between energy storage and transmission infrastructure, our approach simplifies the model by retaining only essential constraints. Moving beyond the scope of centralized optimization scenarios currently considered, our future endeavors will aim at market design for the integration of energy storage as a transmission asset. The market incentive will focus on integrating the varied technological capabilities of energy storage and transmission lines to achieve an optimal mixture.

\bibliographystyle{IEEEtran}
\bibliography{ref.bib}

\end{document}